\documentclass{llncs}
\usepackage{amsmath}
\usepackage{amssymb}
\usepackage{color}
\usepackage{enumitem}
\usepackage{float}
\usepackage{graphicx}
\usepackage{hyperref}
\usepackage{multicol}
\usepackage{setspace}
\usepackage{subfigure}
\usepackage{textcomp}
\usepackage{times}
\usepackage{verbatim}
\newcommand{\cel}[1]{\lceil #1\rceil}

\newcommand{\gdeg}[2]{\mbox{deg}_{#1}(#2)}

\newcommand{\lln}{\log \log n}
\newcommand{\llln}{\log \log \log n}
\newcommand{\PROB}{\mathbf{P}}

\renewenvironment{proof}{\noindent {\bf Proof:}\ }{\qed\par\vskip 2mm}

\floatstyle{ruled}
\newfloat{algo}{htbp}{lop}
\floatname{algo}{Algorithm}

\floatstyle{ruled}
\newfloat{algo2}{hb}{lop}
\floatname{algo2}{Algorithm}

\title{Lessons from the Congested Clique Applied to MapReduce
\thanks{This work is supported in part by National Science Foundation grant
CCF 1318166. E-mail: \texttt{[james-hegeman,sriram-pemmaraju]@uiowa.edu}.}}

\author{James W.~Hegeman\and Sriram V.~Pemmaraju}
\date{}

\institute{Department of Computer Science\\
The University of Iowa\\
Iowa City, Iowa 52242-1419, USA}

\begin{document}
\maketitle

\begin{abstract}
The main results of this paper are (I) a simulation algorithm which, under quite
general constraints, transforms algorithms running on the Congested Clique into
algorithms running in the MapReduce model, and (II) a distributed
$O(\Delta)$-coloring algorithm running on the Congested Clique which has an
expected running time of $O(1)$ rounds, if $\Delta \geq \Theta(\log^4 n)$;
and $O(\llln)$ rounds otherwise. Applying the simulation theorem to the
Congested Clique $O(\Delta)$-coloring algorithm yields an $O(1)$-round
$O(\Delta)$-coloring algorithm in the MapReduce model.

Our simulation algorithm illustrates a natural correspondence between per-node
bandwidth in the Congested Clique model and memory per machine in the MapReduce
model. In the Congested Clique (and more generally, any network in the
$\mathcal{CONGEST}$ model), the major impediment to constructing fast algorithms
is the $O(\log n)$ restriction on message sizes. Similarly, in the MapReduce
model, the combined restrictions on memory per machine and total system memory
have a dominant effect on algorithm design. In showing a fairly general
simulation algorithm, we highlight the similarities and differences between
these models. 
\end{abstract}

\section{Introduction}
\label{sect:Intro}

The $\mathcal{CONGEST}$ model of distributed computation is a synchronous,
message-passing model in which the amount of information that a node can
transmit along an incident edge in one round is restricted to $O(\log n)$ bits
\cite{peleg2000distributed}. As the name suggests, the $\mathcal{CONGEST}$ model
focuses on \textit{congestion} as an obstacle to distributed computation.
Recently, a fair amount of research activity has focused on the design of
distributed algorithms in the $\mathcal{CONGEST}$ model assuming that the
underlying communication network is a \textit{clique}
\cite{berns2012facloc,DolevLenzenPeled,lenzen2013routing,Patt-ShamirTeplitsky}.
Working with such a \textit{Congested Clique} model completely removes from the
picture obstacles that might be due to nodes having to acquire information
from distant nodes (since any two nodes are neighbors), thus allowing us to
focus on the problem of congestion alone. Making this setting intriguing is also
the fact that no non-trivial lower bounds for computation on a Congested Clique
have been proved. In fact, in a recent paper, Lenzen \cite{lenzen2013routing}
showed how to do load-balancing deterministically so as to route up to $n^2$
messages (each of size $O(\log n)$) in $O(1)$ rounds in the Congested Clique
setting, provided each node is the source of at most $n$ messages and the sink
for at most $n$ messages. Thus a large volume of information can be moved around
the network very quickly and any lower-bound approach in the Congested Clique
setting will have to work around Lenzen's routing-protocol result. While Lotker
et al.~\cite{lotker2006distributed} mention overlay networks as a possible
practical application of distributed computation on a Congested Clique, as of
now, research on this model is largely driven by a theoretical interest in
exploring the limits imposed by congestion.

\textit{MapReduce} \cite{DeanGhemavat} is a tremendously popular
parallel-programming framework that has become the tool of choice for
large-scale data analytics at many companies such as Amazon, Facebook, Google,
Yahoo!, etc., as well as at many universities. While the actual time-efficiency
of a particular MapReduce-like implementation will depend on many low-level
technical details, Karloff et al.~\cite{KarloffSuriVassilvitskii} have attempted
to formalize key constraints of this framework to propose a
\textit{MapReduce model} and an associated MapReduce complexity class
($\mathcal{MRC}$). Informally speaking, a problem belongs to $\mathcal{MRC}$ if
it can be solved in the MapReduce framework using: (i) a number of machines that
is substantially sublinear in the input size, i.e., $O(n^{1-\epsilon})$ for
constant $\epsilon > 0$, (ii) memory per machine that is substantially sublinear
in the input size, (iii) $O(\mbox{poly}(\log n))$ number of map-shuffle-reduce
rounds, and (iv) polynomial-time local computation at each machine in each
round. Specifically, a problem is said to be in $\mathcal{MRC}^i$ if it can be
solved in $O(\log^i n)$ map-shuffle-reduce rounds, while maintaining the other
constraints mentioned above. Karloff et al.~\cite{KarloffSuriVassilvitskii} show
that \textit{minimum spanning tree} (MST) is in $\mathcal{MRC}^0$ (i.e., MST
requires $O(1)$ map-shuffle-reduce rounds) on non-sparse instances. Following up
on this, Lattanzi et al.~\cite{LattanziFiltering} show that other problems such
as \textit{maximal matching} (with which the distributed computing community is
very familiar) are also in $\mathcal{MRC}^0$ (again, for non-sparse instances).
We give a more-detailed description of the MapReduce model in Section
\ref{subsect:Models}.

The volume of communication that occurs in a Shuffle step can be quite
substantial and provides a strong incentive to design algorithms in the
MapReduce framework that use very few map-shuffle-reduce steps. As motivation
for their approach (which they call \textit{filtering}) to designing MapReduce
algorithms, Lattanzi et al.~\cite{LattanziFiltering} mention that past attempts
to ``shoehorn message-passing style algorithms into the framework'' have led to
inefficient algorithms. While this may be true for distributed message-passing
algorithms in general, we show in this paper that algorithms designed in the
Congested Clique model provide many lessons on how to design algorithms in the
MapReduce model. We illustrate this by first designing an expected-$O(1)$-round
algorithm for computing a $O(\Delta)$-coloring for a given $n$-node graph with
maximum degree $\Delta \geq \log^4 n$ in the Congested Clique model. We then
simulate this algorithm in the MapReduce model and obtain a corresponding
algorithm that uses a constant number of map-shuffle-reduce rounds to compute an
$O(\Delta)$-coloring of the given graph. While both of these results are new,
what we wish to emphasize in this paper is the \textit{simulation} of
Congested Clique algorithms in the MapReduce model. Our simulation can also be
used to obtain efficient MapReduce-model algorithms for other problems such as
\textit{$2$-ruling sets} \cite{berns2012facloc} for which an
expected-$O(\lln)$-round algorithm on a Congested Clique was recently developed.

\subsection{Models}
\label{subsect:Models}

\paragraph{The Congested Clique Model.} The Congested Clique is a variation on
the more general $\mathcal{CONGEST}$ model. The underlying communication network
is a size-$n$ clique, i.e., every pair of nodes can directly communicate with
each other. Computation proceeds in synchronous rounds and in each round a node
(i) receives all messages sent to it in the previous round; (ii) performs
unlimited local computation; and then (iii) sends a, possibly distinct, message
of size $O(\log n)$ to each other node in the network. We assume that nodes have
distinct IDs that can each be represented in $O(\log n)$ bits. We call this the
\textit{Congested Clique} model.

Our focus in this paper is graph problems and we assume that the input is a
graph $G$ that is a spanning subgraph of the communication network. Initially,
each node in the network knows who its neighbors are in $G$. Thus knowledge of
$G$ is distributed among the nodes of the network, with each node having a
particular local view of $G$. Note that $G$ can be quite dense (e.g., have
$\Omega(n^2)$ edges) and therefore any reasonably fast algorithm for the problem
will have to be ``truly'' distributed in the sense that it cannot simply rely on
shipping off the problem description to a single node for local computation.

\paragraph{The MapReduce Model.} Our description of the MapReduce model borrows
heavily from the work of Karloff et al.~\cite{KarloffSuriVassilvitskii} and
Lattanzi et al.~\cite{LattanziFiltering}. Introduced by Karloff et
al.~\cite{KarloffSuriVassilvitskii}, the MapReduce model is an abstraction of
the popular MapReduce framework \cite{DeanGhemavat} implemented at Google and
also in the popular Hadoop open-source project by Apache.

The basic unit of information in the MapReduce model is a $(key, value)$-pair.
At a high level, computation in this model can be viewed as the application of a
sequence of functions, each taking as input a collection of $(key, value)$-pairs
and producing as output a new collection of $(key, value)$-pairs. MapReduce
computation proceeds in rounds, with each round composed of a map phase,
followed by a shuffle phase, followed by a reduce phase. In the map phase,
$(key, value)$ pairs are processed individually and the output of this phases is
a collection of $(key, value)$-pairs. In the shuffle phase, these
$(key, value)$-pairs are ``routed'' so that all $(key, value)$-pairs with the
same $key$ end up together. In the last phase, namely the reduce phase, each key
and all associated values are processed together. We next describe each of the
three phases in more detail.
\begin{itemize}
\item The computation in the Map phase of round $i$ is performed by a collection
of \textit{mappers}, one per $(key, value)$ pair. In other words, each mapper
takes a $(key, value)$ pair and outputs a collection of $(key, value)$ pairs.
Since each mapper works on an individual $(key, value)$ pair and the computation
is entirely ``stateless'' (i.e., not dependent on any stored information from
previous computation), the mappers can be arbitrarily distributed among
machines. In the MapReduce model, keys and values are restricted to the word
size of the system, which is $\Theta(\log n)$. Because of this restriction, a
mapper takes as input only a constant number of words.
\item In the \textit{Shuffle} phase of round $i$, which runs concurrently with
the Map phase (as possible), key-value pairs emitted by the mappers are moved
from the machine that produced them to the machine which will run the reducer
for which they are destined; i.e., a key-value pair $(k, v)$ emitted by a mapper
is physically moved to the machine which will run the reducer responsible for
key $k$ in round $i$. The
Shuffle phase is implemented entirely by the underlying MapReduce framework 
and we generally ignore the Shuffle phase and treat data movement from one 
machine to another as a part of the Map phase.
\item In the \textit{Reduce} phase of round $i$, 
reducers operate on the collected
key-value pairs sent to them; a reducer is a function taking as input a pair
$(k, \{v_{k,j}\}_j)$, where the first element is a key $k$ and the second is a
multiset of values $\{v_{k,j}\}_j$ which comprises all of the values contained
in key-value pairs emitted by mappers during round $i$ and having key $k$.
Reducers emit a multiset of key-value pairs $\{(k, v_{k,l})\}_l$, where the key
$k$ in each pair is the same as the key $k$ of the input.
\end{itemize}
\noindent
For our purposes, the concepts of a machine and a reducer are interchangeable,
because reducers are allowed to be ``as large'' as a single machine on which
they compute. 

The MapReduce model of Karloff et al.~\cite{KarloffSuriVassilvitskii} tries to
make explicit three key resource constraints on the MapReduce system.
Suppose that the problem input has size $n$ (note that this is \textit{not} referring
to the input size of a particular reducer or mapper).
We assume, as do Karloff et
al.~\cite{KarloffSuriVassilvitskii} and Lattanzi et
al.~\cite{LattanziFiltering}, that memory is measured in $O(\log n)$-bit-sized
words.

\begin{enumerate}
\item Key-sizes and value-sizes are restricted to a $\Theta(1)$ multiple of the word size
of the system. Because of this restriction, a mapper
takes as input only a constant number of words.

\item Both mappers and reducers are restricted to using space consisting of
$O(n^{1-\epsilon})$ words of memory, and time which is polynomial in $n$.

\item The number of machines, or equivalently, the number of reducers, is
restricted to $O(n^{1-\epsilon})$. 
\end{enumerate}

\noindent 
Given these constraints, the goal is to design MapReduce algorithms that run
in very few -- preferably constant -- number of rounds.
For further details on the justifications for these constraints, see
\cite{KarloffSuriVassilvitskii}.

Since our focus is graph algorithms, we can restate the above constraints more
specifically in terms of graph size. Suppose that an $n$-node graph $G = (V, E)$
is the input. Following Lattanzi et al.~\cite{LattanziFiltering}, we assume that
each machine in the MapReduce system has memory $\eta = n^{1+\epsilon}$ for
$\epsilon \geq 0$. Since $n^{1+\epsilon}$ needs to be ``substantially''
sublinear in the input size, we assume that the number of edges $m$ of $G$ is
$\Omega(n^{1+c})$ for $c > \epsilon$.
Thus the MapReduce results in this paper are for non-sparse graphs.

\subsection{Contributions}
\label{subsect:Contributions}

The main contribution of this paper is to show that fast algorithms in the
Congested Clique model 
can be translated via a simulation theorem into fast algorithms in the MapReduce
framework. As a case study, we design a fast graph-coloring algorithm running in
the Congested Clique model and then apply the simulation theorem to this
algorithm and obtain a fast MapReduce algorithm. Specifically, given an $n$-node
graph $G$ with maximum degree $\Delta \geq \log^4 n$, we show how to compute an
$O(\Delta)$-coloring of $G$ in expected $O(1)$ rounds in the Congested Clique model. 
We also present an algorithm for small $\Delta$; for $\Delta < \log^4 n$
we present an algorithm that computes a $\Delta+1$ coloring in 
$O(\llln)$ rounds with high probability on a Congested Clique. 
The implication of this result to the MapReduce model (via the
simulation theorem) is that for any $n$-node graph with $\Omega(n^{1+c})$ edges, for
constant $c > 0$, there is a MapReduce algorithm that runs in $O(1)$
map-shuffle-reduce rounds using $n^{1+\epsilon}$ memory per machine, for
$0 \leq \epsilon < c$ and $n^{c-\epsilon}$ machines. Note that the even when
using $n$ memory per machine and $n^c$ machines the algorithm still takes $O(1)$
rounds. This is in contrast to examples in Lattanzi et
al.~\cite{LattanziFiltering} such as maximal matching which require $O(\log n)$
rounds if the memory per machine is $n$.

The coloring algorithms in both models are new and faster than any known in the
respective models, as far as we know. However, the bigger point of this paper is
the connection between models that are studied in somewhat different
communities.

\subsection{Related Work}
\label{subsect:Related}

The earliest interesting algorithm in the Congested Clique model is an
$O(\lln)$-round deterministic algorithm to compute a minimum spanning tree (MST), due
to Lotker et al.~\cite{lotker2006distributed}. Gehweiler et
al.~\cite{GehweilerSPAA2006} presented a random $O(1)$-round algorithm in the
Congested Clique model that produced a constant-factor approximation algorithm
for the \textit{uniform} metric facility location problem. Berns et
al.~\cite{berns2012facloc,berns2012arxiv} considered the more-general
non-uniform metric facility location in the Congested Clique model and presented
a constant-factor approximation running in expected $O(\lln)$ rounds. Berns et
al.~reduce the metric facility location problem to the problem of computing a
$2$-ruling set of a spanning subgraph of the underlying communication network
and show how to solve this in $O(\lln)$ rounds in expectation. In 2013, Lenzen
presented a routing protocol to solve a problem called an
\textit{Information Distribution Task} \cite{lenzen2013routing}. The setup for
this problem is that each node $i \in V$ is given a set of $n' \leq n$ messages,
each of size $O(\log n)$, $\{m_i^1, m_i^2, \ldots, m_i^{n'}\}$, with
destinations $d(m_i^j) \in V$, $j \in \{1, 2, \ldots, n'\}$. Messages are
globally lexicographically ordered by their source $i$, destination $d(m_i^j)$,
and $j$. Each node is also the destination of at most $n$ messages. Lenzen's
routing protocol solves the Information Distribution Task in $O(1)$ rounds.

Our main sources of reference on the MapReduce model and for graph algorithms in
this model are the work of Karloff et al.~\cite{KarloffSuriVassilvitskii} and
Lattanzi et al.~\cite{LattanziFiltering} respectively. Besides these, the work
of Ene et al.~\cite{EneClusteringMapReduce} on algorithms for clustering in
MapReduce model and the work of Kumar et al.~\cite{KumarGreedyMapReduce} on
greedy algorithms in the MapReduce model are relevant.

\section{Coloring on the Congested Clique}
\label{sect:CongestedClique}

In this section we present an algorithm, running in the Congested Clique model,
that takes an $n$-node graph $G$ with maximum degree $\Delta$ and computes an
$O(\Delta)$-coloring in expected $O(\llln)$ rounds. In fact, for high-degree
graphs, i.e., when $\Delta \geq \log^4 n$, our algorithm computes an
$O(\Delta)$-coloring in $O(1)$ rounds. This algorithm, which we call Algorithm
\textsc{HighDegCol}, is the main contribution of this section. For graphs with
maximum degree less than $\log^4 n$ we appeal to an already-known coloring
algorithm that computes a $(\Delta + 1)$ coloring in $O(\log \Delta)$ rounds
and then modify its implementation so that it runs in $O(\llln)$ rounds
on a Congested Clique.

We first give an overview of Algorithm \textsc{HighDegCol}. The reader is
advised to follow the pseudocode given in Algorithm \ref{alg:DistColorCC} as
they read the following. The algorithm repeatedly performs a simple random trial
until a favorable event occurs. Each trial is independent of previous trials.
The key step of Algorithm \textsc{HighDegCol} is that each node picks a
\textit{color group} $k$ from the set $\{1, 2, \ldots, \cel{\Delta / \log n}\}$
independently and uniformly at random (Step 4). We show (in Lemma
\ref{alg:DistColorCC}) that the expected number of edges in the graph $G_k$ induced by
nodes in color group $k$ is at most $O(\frac{n \log^2 n}{\Delta})$. Of course, some of
the color groups may induce far more edges and so we define a \textit{good}
color group as one that has at most $n$ edges. The measure of whether the
random trial has succeeded is the number of good color groups. If most of the color
groups are good, i.e., if at most $2 \log n$ color groups are not good then the
random trial has succeeded and we break out of the loop. We then transmit each
graph induced by a good color group to a distinct node in constant rounds using
Lenzen's routing scheme \cite{lenzen2013routing} (Step 11). Note that this is
possible because every good color group induces a graph that requires $O(n)$
words of information to completely describe. Every node that receives a graph
induced by a good color group locally computes a proper coloring of the graph
using one more color than the maximum degree of the graph it receives (Step 12).
Furthermore, every such coloring in an iteration employs a distinct palette of
colors. Since there are very few color groups that are not good, we are able to
show that the residual graph induced by nodes not in good color groups has
$O(n)$ edges. As a result, the residual graph can be communicated in its
entirety to a single node for local processing. This completes the coloring of
all nodes in the graph.

\begin{algo}[h]
\textbf{Input:} An $n$-node graph $G = (V, E)$, of maximum degree $\Delta$\\
\textbf{Output:} A proper node-coloring of $G$ using $O(\Delta)$ colors
{\small
\begin{tabbing}
......\=a....\=b....\=c....\=d....\=e....\=f....\=g....\=h....\=i....\=j...\kill
1.\>Each node $u$ in $G$ computes and broadcasts its degree to every other node
$v$ in $G$.\\
2.\>\textbf{If} $\Delta \leq \log^4 n$ \textbf{then} use Algorithm
\textsc{LowDegCol} instead.\\
3.\>\textbf{while} $true$ \textbf{do}\\
4.\>\>Each node $u$ chooses a \textit{color group} $k$ from the set
$\{1, 2, \ldots, \cel{\Delta / \log n}\}$ independently\\ 
\>\>\>and uniformly at random.\\
5.\>\>Let $G_k$ be the subgraph of $G$ induced by nodes of color group $k$.\\
6.\>\>Each node $u$ sends its choice of color group to all neighbors in
$G$.\\
7.\>\>Each node $u$ computes its degree within its own color-group graph
$G_{k_u}$ and sends its\\ 
\>\>\>color group and degree within color group to node 1.\\
8.\>\>Node 1, knowing the partition of $G$ into color groups and also
knowing the degree of\\
\>\>\>every node $u$ ($u \in G_k$) within the induced subgraph $G_k$, can compute the
number\\
\>\>\>of edges in $G_k$ for each $k$. Thus node 1 can determine which color-group graphs $G_k$\\ 
\>\>\>are \textit{good}, i.e., have at most $n$ edges.\\
9.\>\>If at most $2 \log n$ color-group graphs are not good, node 1 broadcasts a ``break''\\
\>\>\>message to all nodes causing them to \textbf{break} out of loop;\\
\>\textbf{endwhile}\\
10.\>Node 1 informs every node $u$ in a good group of the fact that $u$'s color 
group is good\\
11.\>Using Lenzen's routing protocol, distribute all information about all good
color-group\\ 
\>\>graphs $G_k$ to distinct nodes of $G$.\\
12.\>For each good $G_k$, the recipient of $G_k$ computes a coloring of $G_k$
using $\Delta(G_k)+1$ colors.\\
\>\>The color palette used for each $G_k$ is distinct.\\
13.\>The residual graph $\overline{G}$ of uncolored nodes has size $O(n)$ with
high probability, and can thus\\
\>\>be transmitted to a single node (for local proper coloring) in
$O(1)$ rounds.\\
14.\>Each node that locally colors a subgraph informs each  node in the subgraph
the color it has\\
\>\>been assigned.
\end{tabbing}}
\caption{\textsc{HighDegCol}}
\label{alg:DistColorCC}
\end{algo}

We now analyze Algorithm \textsc{HighDegCol} and show that (i) it terminates in
expected-$O(1)$ rounds and (ii) it uses $O(\Delta)$ colors. Subsequently, we
discuss an $O(\llln)$ algorithm to deal with the small $\Delta$ case.

\begin{lemma}
For each $k$, the expected number of edges in $G_k$ is
$\frac{n \log^2 n}{2 \Delta}$.
\label{lemma:group_size}
\end{lemma}
\begin{proof}
Consider edge $\{u, v\}$ in $G$. The probability that both $u$ and $v$ choose
color group $k$ is at most
$\frac{\log n}{\Delta} \cdot \frac{\log n}{\Delta} = \frac{\log^2 n}{\Delta^2}$.
Since $G$ has at most $\frac{1}{2} \Delta \cdot n$ edges, the expected number of
edges in $G_k$ is at most $\frac{n \log^2 n}{2 \Delta}$.
\end{proof}

\begin{lemma}
The expected number of color-group graphs $G_k$ having more than $n$ edges is at
most $\log n$.
\label{lemma:bad_groups}
\end{lemma}
\begin{proof}
By Lemma \ref{lemma:group_size} and Markov's inequality, the probability that
color group $k$ has more than $n$ edges is at most
$\frac{n \log^2 n}{2 \Delta \cdot n} = \frac{\log^2 n}{2 \Delta}$. Since there
are $\cel{\Delta / \log n}$ groups, the expected number of $G_k$ having more
than $n$ edges is bounded above by
$2 \frac{\Delta}{\log n} \cdot \frac{\log^2 n}{2 \Delta} = \log n$.
\end{proof}

\begin{lemma}
With high probability, every color group has $\frac{5 n \log n}{\Delta}$ nodes.
\label{lemma:no_large_groups}
\end{lemma}
\begin{proof}
The number of color groups is $\cel{\Delta / \log n}$. Thus, for any $k$, the
expected number of nodes in $G_k$, denoted $|V(G_k)|$, is at most
$n \cdot \frac{\log n}{\Delta}$. An application of a Chernoff bound then gives,
for each $k$,
\[\PROB\left(|V(G_k)| > 5 n \cdot \frac{\log n}{\Delta}\right) \leq 
2^{-5 n \cdot \frac{\log n}{\Delta}} < 2^{-5 \log n} = \frac{1}{n^5}\]
Taking the union over all $k$ completes the proof.
\end{proof}

\begin{lemma}
With high probability, no node $u$ in $G$ has more than $5 \log n$ neighbors in
any color group.
\label{lemma:degree_bound}
\end{lemma}
\begin{proof}
Node $u$ has maximum degree $\Delta$, so for any $k$, the expected number of
neighbors of $u$ which choose color group $k$ is bounded above by $\log n$.
Therefore, applying a Chernoff bound gives
\[\PROB\left(|N(u) \cap G_k| > 5 \log n\right) \leq 2^{-5 \log n} = \frac{1}{n^5}\]
Taking the union over all $k$ and $u$ shows that, with probability at least
$1 - \frac{1}{n^3}$, the assertion of the lemma holds.
\end{proof}

\begin{lemma}
The residual graph $\overline{G}$, induced by groups that are good, has $O(n)$ edges, with high probability.
\label{lemma:residual_graph_size}
\end{lemma}
\begin{proof}
The residual graph $\overline{G}$ is a graph induced by at most $2 \log n$ color
groups, since the algorithm is designed to terminate only when it has performed a trial resulting in
at most $2 \log n$ groups that are not good.
With high probability, no node $u$ in $\overline{G}$ has more than
$5 \log n$ neighbors in any of the (at most) $2 \log n$ color groups that make up $\overline{G}$, so
therefore with high probability no node $u$ has degree greater than $10 \log^2 n$ in $\overline{G}$.
Since $\overline{G}$ has at most $(2 \log n) \cdot \frac{5 n \log n}{\Delta}$
nodes with high probability, it follows that the number of edges in
$\overline{G}$ is at most
\[(2 \log n) \cdot \frac{5 n \log n}{\Delta} \cdot 10 \log^2 n =
\frac{100 n \log^4 n}{\Delta}\]
which is $O(n)$ when $\Delta \geq \log^4 n$.
\end{proof}

\begin{lemma}
Algorithm \textsc{HighDegCol} runs in a constant number of rounds, in
expectation.
\label{lemma:constantCommunication}
\end{lemma}
\begin{proof}
By Lemma \ref{lemma:bad_groups} and Markov's inequality, the expected number of
color-group partitioning attempts required before the number of ``bad'' color
groups (i.e., color groups whose induced graphs $G_k$ contain more than $n$
edges) is less than or equal to $2 \log n$ is two. It is easy to verify that
each iteration of the \textbf{while}-true loop requires $O(1)$ rounds of
communication.

When $\Delta \geq \log^4 n$, the residual graph $\overline{G}$ is of size $O(n)$
with high probability, and can thus be communicated in its entirety to a single
node in $O(1)$ rounds. That single node can then color $\overline{G}$
deterministically using $\Delta + 1$ colors and then inform every node of
$\overline{G}$ of its determined color in one further round.
\end{proof}

\begin{lemma}
\label{lemma:deltaColors}
Algorithm \textsc{HighDegCol} uses $O(\Delta)$ colors.
\end{lemma}
\begin{proof}
A palette of size $O(\log n)$ colors suffices for each good color group because we showed in Lemma
\ref{lemma:degree_bound} that with high probability the maximum degree in any
color group is $5 \log n$. 
Since there are a total of $\cel{\Delta / \log n}$ color groups and we use a distinct
palette of size $O(\log n)$ for each good color group, we use a total of
$O(\Delta)$ colors for the good color groups. 
The residual graph induced by not-good color groups
is colored in the last step and it requires an additional $O(\Delta)$ colors.
\end{proof}

\subsection{Coloring low-degree graphs}
\label{subsection:lowDegreeGraphs}

Now we describe an algorithm that we call \textsc{LowDegCol} that, given an
$n$-node graph $G$ with maximum degree $\Delta < \log^4 n$, computes a proper
$(\Delta + 1)$-coloring with high probability in $O(\llln)$ rounds in the Congested Clique model. 
The algorithm has two stages.
The first stage of the algorithm is based on the simple, natural, randomized coloring
algorithm first analyzed by Johannson \cite{Johansson} and more recently by
Barenboim et al.~\cite{BEPS12}. Each node $u$ starts with a color palette
$C_u = \{1, 2, \ldots, \Delta + 1\}$. In each iteration, each as-yet uncolored
node $u$ makes a tentative color choice $c(u) \in C_u$ by picking a color from
$C_u$ independently and uniformly at random. If no node in $u$'s neighborhood
picks color $c(u)$ then $u$ colors itself $c(u)$ and $c(u)$ is deleted from the
palettes of all neighbors of $u$. Otherwise, $u$ remains uncolored and
participates in the next iteration of the algorithm.
We call one such iteration \textsc{RandColStep}.
Barenboim et al.~\cite{BEPS12} show (as part of the proof of Theorem
5.1) that if we executed $O(\log \Delta)$ iterations of \textsc{RandColStep}, then
with high probability the nodes that remain uncolored induce connected components of size
$O(\mbox{poly}(\log n))$. 
Since we are evaluating a situation in which $\Delta < \log^4 n$, this translates to
using $O(\lln)$ iterations of \textsc{RandColStep} to reach a state with small connected components.
Now notice that this algorithm uses only the edges of $G$ -- the graph being colored -- for
communication.
By utilizing the entire bandwidth of the underlying clique communication network, it is possible 
to speed up this algorithm significantly and get it to complete in $O(\llln)$ rounds.
The trick to doing this is to rapidly gather, at each node $u$, all information needed by node $u$
to execute the algorithm locally.
We make this precise further below.

Once we execute $O(\llln)$ iterations of \textsc{RandColStep} and
all connected components induced by as-yet uncolored nodes become polylogarithmic in size, then 
Stage 2 of the algorithm begins.
In this stage, first each connected component is gathered at a node; we show how to accomplish this in
$O(\llln)$ rounds by appealing to the deterministic MST algorithm on a
Congested Clique due to Lotker et al.~\cite{lotker2006distributed}.
Then each connected component of uncolored nodes is 
shipped off to a distinct node and is 
locally (and independently) colored using $\Delta+1$
colors.


We start by developing Stage 1 first.
Suppose that for some constants $c_1, c_2, c_3$, $T < c_1 \log\log n$ iterations of \textsc{RandColStep} 
are needed before all connected components induced by uncolored nodes have size at most
$c_2 \cdot \log^{c_3} n$ with probability at least $1 - 1/n$.
Let $G_L$ denote a labeled version of graph $G$ in which each node $u$ is
labeled $(\texttt{ID}_u, \texttt{RS}_u)$,
where $\texttt{ID}_u$
is the $O(\log n)$-bit ID of node $u$ and 
$\texttt{RS}_u$ is a random bit string of length $T \cdot \lceil \log\Delta \rceil$.
For integer $k \ge 0$ and node $u \in V$, let $B(u, k)$ denote the set of all
nodes within $k$ hops of $u$ in $G$.
The following lemma shows that it is quite helpful if each node $u$ knew $G_L[B(u, T)]$,
the subgraph of the labeled graph $G_L$ induced by nodes in $B(u, T)$.

\begin{lemma}
\label{lemma:afterBalls}
Suppose that each node $u \in V$ knows $G_L[B(u, T)]$.
Then each node $u$ can locally compute a color $c(u) \in \{\bot\} 
\cup \{1, 2, \ldots, \Delta+1\}$ such that (i) nodes not colored $\bot$
induce a properly colored subgraph and (ii) nodes colored $\bot$ induce
connected components whose size is bounded above by $c_2 \log^{c_3} n$
with probability at least $1 - 1/n$.
\end{lemma}
\begin{proof}
With respect to the execution of iterations of \textsc{RandColStep}, the \textit{state} of a node $u$ is its current 
color palette $C_u$ and its current color choice $c(u)$.
If $c(u) = \bot$, then $u$ has not colored itself; otherwise, $c(u)$ is a permanently assigned
color that node $u$ has given itself.
To figure out the state of node $u$ after $T$ iterations of \textsc{RandColStep}, it suffices to
know (i) the state of $u$ and its neighbors after $T-1$ iterations of \textsc{RandColStep} and
(ii) at most $\lceil \log \Delta\rceil$ random bits associated with each of these nodes so that their random color choices
in iteration $T$ can be determined.
Stated differently, it suffices to know (i) the subgraph $G_L[B(u, 1)]$ and (ii) the state of each node in $B(u, 1)$ after $T-1$ iterations
of \textsc{RandColStep}.
This in turn can be computed from (i) the subgraph $G_L[B(u, 2)]$ and (ii) the state of all nodes in $B(u, 2)$ after $T-2$ iterations
of \textsc{RandColStep}.
Continuing inductively, we conclude that in order to know the state of node $u$ after $T$ iterations
of \textsc{RandColStep}, it suffices to know $G_L[B(u, T)]$, where each node $v$ in $B(u, T)$ is labeled with 
an $(\texttt{ID}_v, \texttt{RS}_v)$-pair, where $\texttt{RS}_v$ is a random bit string of length $T \cdot \lceil \log\Delta \rceil$.
\end{proof}

\noindent
Now we focus on the problem of each node gathering 
$G_L[B(u, T)]$ and show that this problem can be solved 
in $O(\llln)$ rounds, given that $T = O(\lln)$ and $\Delta < \log^4 n$.

\begin{lemma}
\label{lemma:ballGrowing}
There is a Congested Clique algorithm running on an $n$-node input graph $G$
with maximum degree $\Delta < \log^4 n$ that terminates in $O(\llln)$ rounds
at the end of which, every node $u$ knows 
$G_L[B(u, T)]$.
\end{lemma}
\begin{proof}
The algorithm starts with each node $u$ broadcasting its degree in $G$ to all nodes in $V$.
This enables every node to locally compute $\Delta$ and also a random bit string $\texttt{RS}_u$ of
length $T \cdot \lceil \log \Delta \rceil$.
After computing $\texttt{RS}_u$, each node $u$ sends to each neighbor in $G$ the pair $(\texttt{ID}_u, \texttt{RS}_u)$.
Now each node $u$ is in possession of the collection of $(\texttt{ID}_v, \texttt{RS}_v)$-pairs for all
neighbors $v$.
Each node $u$ now has a goal of sending this collection to every neighbor.
Note that the total volume of information that $u$ wishes to send out is bounded above by $\Delta^2$ (measured
in $O(\log n)$-sized words).
Also, each node $u$ is the destination for at most $\Delta^2$ words.
Since $\Delta^2 = o(n)$, 
using Lenzen's routing protocol \cite{lenzen2013routing}, each node can successfully send its entire collection of
$(\texttt{ID}, \texttt{RS})$-pairs to all neighbors in constant rounds.
Based on this received information, each node $u$ can construct $G_L[B(u, 1)]$.

Proceeding inductively, suppose that each node $u$ has gathered $G_L[B(u, t)]$, where $1 \le t < T$.
We now show that in an additional constant rounds, $u$ can gather $G_L[B(u, 2t)]$.
First note that $|B(u, t)| \le \Delta^{t+1}$ for any node $u \in V$.
Therefore, $G_L[B(u, t)]$ can be completely described using $O(\Delta^{t+2})$ words
of information.
In order to compute $G_L[B(u, 2t)]$, each node $u$ sends $G_L[B(u, t)]$ to each node
in $B(u, t)$.
A node $u$, on receiving $G_L[B(v, t)]$ for all nodes $v$ in $B(u, t)$, can perform a local
computation to determine $G_L[B(u, 2t)]$. 
Note that the total volume of information that $u$ needs to send out during this communication is $O(\Delta^{2t+3})$ words.
By a symmetric reasoning, each node $u$ is the destination for at most $O(\Delta^{2t+3})$ words of information.
Since $\Delta < \log^4 n$ and $t < T = O(\lln)$, $\Delta^{2t+3} = o(n)$ and therefore using Lenzen's routing
protocol, each node $u$ can send $G_L[B(u, t)]$ to each node in $B(u, t)$ in constant rounds.

Since the goal of the algorithm is for each node $u$ to learn $G_L[B(u, T)]$, where $T = O(\lln)$,
it takes $O(\llln)$ iterations of the above described inductive procedure to reach
this goal. 
The result follows from the fact that each iteration involves a constant number of communication rounds.
\end{proof}

An immediate consequence of Lemmas \ref{lemma:afterBalls} and
\ref{lemma:ballGrowing} is that
there is a Congested Clique algorithm running on an $n$-node input graph $G$
with maximum degree $\Delta < \log^4 n$ that terminates in $O(\llln)$ rounds
at the end of which, every node $u$ has assigned itself a color $c(u) \in \{\bot\} \cup \{1, 2, \ldots, \Delta+1\}$
such that (i) nodes not colored $\bot$
induce a properly colored subgraph and (ii) nodes colored $\bot$ induce
connected components whose size is bounded above by $O(\mbox{poly}(\log n))$
with probability at least $1 - 1/n$.
This brings us to Stage 2 of our algorithm.
The first task in this stage is to distribute information about uncolored nodes (i.e., nodes $u$ with $c(u) = \bot$)
such that each connected component in the subgraph induced by uncolored nodes ends up at a 
node in the network.
To perform this task in $O(\llln)$ rounds, we construct a complete, edge-weighted graph in which an
edge $\{u, v\}$ has weight $w(u, v) = 1$ if $\{u, v\} \in E$ and $c(u) = c(v) = \bot$ and 
has weight $n$ otherwise.
Thus, edges in the subgraph of $G$ induced by uncolored nodes have weight 1 and edges connecting all other pairs
of nodes have weight $n$.
This complete, edge-weighted graph serves as an input to the MST algorithm of Lotker et al.
Note that this input is distributed across the network with each node having knowledge of the weights 
of all $n-1$ edges incident on it.
Also note that this knowledge can be acquired by all nodes after just one round of communication.
As mentioned earlier, the Lotker et al.~MST algorithm runs in $O(\lln)$ rounds.
Since we are not interested in computing an MST, but only in identifying connected
components, we do not have to run the Lotker et al.~algorithm to completion.

The Lotker et al.~algorithm runs in phases, taking constant number of communication rounds
per phase. At the end of phase $k \ge 0$, the algorithm has computed a partition 
$\mathcal{F}^k = \{F_1^k, F_2^k, \ldots, F_{m_k}^k\}$ of the nodes of $G$ into
\textit{clusters}.
Furthermore, for each cluster $F \in \mathcal{F}^k$, the algorithm has computed a spanning
tree $T(F)$.
The correctness of the algorithm is ensured by the fact that each tree $T(F)$ is a subgraph of the MST.
It is worth noting that every node in the network knows the partition $\mathcal{F}^k$ and the collection
$\{T(F) \mid F \in \mathcal{F}^k\}$ of trees.
Suppose that the minimum size cluster in $\mathcal{F}^k$ has size $N$.
The $O(\lln)$ running time of the Lotker et al.~algorithm arises from the fact that in each phase the algorithm
merges clusters and at the end of Phase $k+1$ the smallest cluster in $\mathcal{F}^{k+1}$ has size at
least $N^2$.
Thus the size of the smallest cluster ``squares'' in each phase and therefore it takes $O(\lln)$ rounds
to get to the stage where the smallest cluster has size $n$, at which point there is only one cluster
$F$ and $T(F)$ is the MST.

We are interested in executing $T$ phases of the Lotker et al.~algorithm so that the size of the smallest
cluster in $\mathcal{F}^T$ is at least the size of the largest connected component induced by uncolored nodes.
Since the size of the largest connected component in the graph induced by uncolored nodes is $O(\mbox{poly}(\log n))$,
it takes only $T = O(\llln)$ phases to reach such a stage.
Let $\mathcal{F}^T = \{F_1^T, F_2^T, \ldots, F_{m}^T\}$ be the partition of the nodes of $G$ into clusters
at the end of $T$ phases of the Lotker et al.~algorithm.

\begin{lemma}
Let $C$ be a connected component in the subgraph induced by uncolored nodes.
Then $C \subseteq F^T_i$ for some $i$. 
\end{lemma}
\begin{proof}
To obtain a contradiction suppose that $C \cap F^T_i \not= \emptyset$ and $C \cap F^T_j \not= \emptyset$ for some $1 \le i \not= j \le m$.
Then there is an edge of weight 1 connecting a node in $F^T_i$ and a node in $F^T_j$.
Since $|F^T_i| \ge |C|$, the tree $T(F^T_i)$ contains an edge of weight $n$.
Thus at some point in the Lotker et al.~algorithm, it chose to merge clusters
using an edge of weight $n$ when it could have used an edge of weight 1.
This contradicts the behavior of the Lotker et al.~algorithm.
\end{proof}

The rest of Stage 2 is straightforward.
One node, say $u^*$, considers each $F \in \mathcal{F}^T$ and deletes all edges of weight $n$ from $T(F)$.
This will result in $F$ splitting up into smaller clusters; these clusters are the connected components
of the subgraph of $G$ induced by uncolored nodes.
Note that at this point we think of a connected component as simply a subset of nodes.
Node $u^*$ then ships off each connected component to a distinct node, possibly the node
with the smallest ID in that component.
This takes constant number of rounds via the use of Lenzen's routing protocol.
Suppose that a node $u$ has received a connected component $C$.
Node $u$ then contacts the nodes in $C$ to find out (i) all edges connecting pairs of nodes
in $C$, and (ii) the current palettes $C_v$ for each node $v \in C$.
Since $|C|$ is polylogarithmic in size and $\Delta < \log^4 n$, it is easy to see that all of this
information requires polylogarithmic number of bits to represent and therefore can be communicated
to $u$ in constant number of rounds via Lenzen's routing protocol.
Node $u$ then colors each node $v \in C$ using a color from its palette $C_v$ such that 
the graph induced by $C$ is properly colored.
This completes Stage 2 and we have a $(\Delta+1)$-coloring of $G$.

\begin{lemma}
Given an $n$-node graph $G$ with maximum degree $\Delta \leq \log^4 n$,
Algorithm \textsc{LowDegCol} computes a proper $(\Delta+1)$-coloring in
$O(\llln)$ rounds in the Congested Clique model.
\label{lemma:lowDegCol}
\end{lemma}

\noindent Combining Lemmas \ref{lemma:constantCommunication} and
\ref{lemma:deltaColors} along with Lemma \ref{lemma:lowDegCol} gives the
following theorem.

\begin{theorem}
Given an $n$-vertex input graph $G = (V, E)$ with maximum degree
$\Delta \geq \log^4 n$, Algorithm \textsc{HighDegCol} computes an
$O(\Delta)$-coloring in $O(1)$ rounds (in expectation) in the Congested Clique
model. For arbitrary $\Delta$, an $O(\Delta)$-coloring can be computed in
$O(\llln)$ rounds in expectation in the Congested Clique model.
\label{theorem:CCColoring}
\end{theorem}

\section{MapReduce Algorithms from Congested Clique Algorithms}
\label{sect:MRfromCC}

In this section, we prove a \textit{simulation} theorem establishing that
Congested Clique algorithms (with fairly weak restrictions) 
can be efficiently implemented in
the MapReduce model. The simulation ensures that a Congested Clique algorithm
running in $T$ rounds can be implemented in $O(T)$ rounds (more precisely,
$3 \cdot T + O(1)$ rounds) in the MapReduce model, if certain communication and
``memory'' conditions are met. The technical details of this simulation are
conceptually straightforward, but the details are a bit intricate.

We will now precisely define restrictions that we need to place on Congested Clique 
algorithms in order for the simulation theorem to go through.
We assume that each node in the Congested Clique possesses a word-addressable memory 
whose words are indexed by the natural numbers. For an algorithm $\mathcal{A}_{CC}$ running in the
Congested Clique, let $I_u^{(j)} \subset \mathbb{N}$ be the set of memory
addresses \textit{used} by node $u$ during the local computation in round $j$
(not including the sending and receipt of messages).

After local computation in each round, each node in the Congested Clique may
send (or not send) a distinct message of size $O(\log n)$ to each other node in
the network. In defining notation, we
make a special distinction for the case where a node $u$ sends in the
\textit{same} message to every other node $v$ in a particular round; i.e., node
$u$ sends a \textit{broadcast} message. The reason for this distinction is that
broadcasts can be handled more efficiently on the receiving end in the MapReduce
framework than can distinct messages sent by $u$. 
Let $m_{u,v}^{(j)}$ denote a message sent by node $u$ to node $v$ in round $j$
and let $D_u^{(j)} \subseteq V$ be the set of destinations of messages sent
by node $u$ in round $j$.
Let $M_u^{(j)} = \{m_{u,v}^{(j)} \colon v \in D_u^{(j)} \subset V\}$ be the set
of messages \textit{sent} by node $u$ in round $j$ of algorithm
$\mathcal{A}_{CC}$, except let
$M_u^{(j)} = \emptyset$ if $u$ has chosen to broadcast a message $b_{u}^{(j)}$
in round $j$. Similarly, let
$\overline{M}_u^{(j)} = \{m_{v,u}^{(j)} \colon u \in D_v^{(j)} \text{ and } v \text{ is not broadcasting in round $j$}\}$
be the set of messages \textit{received} by node $u$ in round $j$, except that
we exclude messages $b_v^{(j)}$ from nodes $v$ that have chosen to broadcast in
round $j$. 
We say that $\mathcal{A}_{CC}$, running on an $n$-node
Congested Clique, is $(K, N)$-\textit{lightweight} if
\begin{itemize}
\item[(i)] for each round $j$ (in the Congested Clique),
$\sum_{u \in V} (|\overline{M}_u^{(j)}| + |I_u^{(j)}|) = O(K)$;
\item[(ii)] there exists a constant $C$ such that for each round $j$ and for
each node $u$, $I_u^{(j)} \subseteq \{1, 2, \ldots, \cel{C \cdot N}\}$; and
\item[(iii)] each node $u$ performs only polynomial-time local computation in
each round.
\end{itemize}
In plain language: no node uses more than $O(N)$ memory for local computation
during a round;
the total amount of memory that all nodes use and the total volume of messages
nodes receive in any round is bounded by $O(K)$.
Regarding condition (iii), traditional models of distributed computation such as
the $\mathcal{CONGEST}$ and $\mathcal{LOCAL}$ models allow nodes to perform
arbitrary local computation (e.g., taking exponential time), but since the
MapReduce model requires mappers and reducers to run in polynomial time, we need
this extra restriction.

\begin{theorem}
Let $\epsilon$, $c$ satisfy $0 \leq \epsilon \leq c$, and let $G = (V, E)$ be a
graph on $n$ vertices having $O(n^{1+c})$ edges. If $\mathcal{A}_{CC}$ is a
$(n^{1+c}, n^{1+\epsilon})$-\textit{lightweight} Congested Clique-model
algorithm running on input $G$ in $T$ rounds, then $\mathcal{A}_{CC}$ can be
implemented in the MapReduce model with $n_r = n^{c-\epsilon}$ machines and
$m_r = \Theta(n^{1+\epsilon})$ (words of) memory per machine such that the
implementation runs in $O(T)$ Map-Shuffle-Reduce rounds on $G$.
\label{theorem:MRfromCC}
\end{theorem}
\begin{proof}
The simulation that will prove the above theorem contains
two stages: the \textit{Initialization} stage and the \textit{Simulation}
stage. In the Initialization stage, the input to the MapReduce
system is transformed from the assumed format (an unordered list of edges and
vertices of $G$) into a format in which each piece of information, be it an
edge, node, or something else, that is associated with a node of $G$ is gathered at a single
machine.
After this gathering of associated information has been completed, the MapReduce system can emulate
the execution of the Congested Clique algorithm.


\noindent\textbf{Initialization} stage. Input (in this case, the graph $G$) in the
MapReduce model is assumed to be presented as an unordered sequence of tuples of
the form $(\varnothing, u)$, where $u$ is a vertex of $G$, or
$(\varnothing, (u, v))$, where $(u, v)$ is an edge of $G$.
The goal of the Initialization stage is to partition the input $G$ among
the $n_r$ reducers such that each reducer $r$ receives a subset 
$P_r \subseteq V$ and all edges $E_r$ incident on nodes in $P_r$ such that
$|P_r| + |E_r|$ is bounded above by $O(n^{1+\epsilon})$.
This stage can be seen as consisting of two tasks: (i) every reducer $r$ learns
the degree $\gdeg{G}{u}$ of every node $u$ in $G$ and (ii) every reducer
computes a partition (the same one) given by the partition function 
$F_0 : V \longrightarrow \{1, 2, \ldots, n_r\}$, defined by
$$F_0(x) = \begin{cases} 1, & \mbox{if } x = 1\\ F_0(x-1), 
& \mbox{if } \sum_{v \in L(x)} \gdeg{G}{v} \leq n^{1+\epsilon},\\
F_0(x-1)+1, & \mbox{otherwise}
\end{cases}
$$
Here $L(x) = \{j < x: F_0(j) = F_0(x-1)\}$.
All nodes in the same group in the partition are mapped to the same value by $F_0$ and will be assigned to a 
single reducer.
Since the degree of each node is bounded above by $n$, it is easy to see that 
for any $r \in \{1, 2, \ldots, n_r\}$, $F_0^{-1}(r)$ is a subset of nodes
of $G$ such that $|F_0^{-1}(r)| + \sum_{u\in F_0^{-1}(r)} \gdeg{G}{u}$
is $O(n^{1+\epsilon})$.
Each of the two tasks mentioned above can be implemented in a (small) constant number of
MapReduce rounds as follows.

\begin{itemize}
\item\textbf{Map 1:} In Map phase 1, for each tuple $(\varnothing, u)$, a
mapper chooses a random reducer $r$ and emits the tuple $(r, u)$. For each tuple
$(\varnothing, (u, v))$, a mapper again chooses a random reducer $r$ and emits
the tuple $(r, (u, v))$. Because the reduce keys are chosen at random, with high
probability (actually, exponentially high probability) each reducer in Reduce
phase 1 will receive $O(n^{1+\epsilon})$ tuples.

\item\textbf{Reduce 1:} In Reduce phase 1, a reducer $r$ receives tuples whose
values consist of some collection $P_r \subseteq V$ of vertices and some
collection $E_r \subseteq E$ of edges of $G$. For each value consisting
of a vertex $u$, a reducer $r$ re-emits the tuple $(r, u)$, and for each value
consisting of an edge $(u, v)$, reducer $r$ re-emits the tuple$(r, (u, v))$. In
addition, a reducer $r$ emits, for each vertex $u$ such that reducer $r$
received an edge $(u, v)$ or $(v, u)$, a tuple $(r, u, d_{r,u})$, where
$d_{r,u}$ is total number of edges received by reducer $r$ containing $u$.
(In other words, $d_{r,u}$ is the partial degree of $u$ seen by reducer $r$.)

\item\textbf{Map 2:} In Map phase 2, mappers again load-balance tuples
containing vertices or edges as values across the reducers uniformly at random
(an action which is successful w.h.p.), as in Map phase 1. In addition, when a
mapper processes a tuple of the form $(r, u, d_{r,u})$, it emits the tuple
$((u \bmod n_r), u, d_{r,u})$. Here $u \bmod n_r$ refers to the reduction of the
identifier of node $u$ modulo the number of reducers, $n_r$. There are at most
$n \cdot n_r = O(n^{1+c-\epsilon})$ such tuples, and thus (i) each reducer is
the destination of $O(n)$ such tuples
(of the form $((u \bmod n_r), u, d_{r,u})$); and (ii) all tuples containing a
partial degree sum of node $u$ among their values are given the same key
and thus sent to the same reducer during the second MapReduce round.

\item\textbf{Reduce 2:} In Reduce phase 2, a reducer $r$ again re-emits tuples
$(r, u)$ and $(r, (u, v))$ for each vertex or edge received as a value. For
tuples of the form $(r, u, d_{r',u})$, reducer $r$ aggregates the partial degree
sums of $u$ to compute the full degree $\gdeg{G}{u}$ of $u$ in $G$, and emits
the tuple $(r, u, \gdeg{G}{u})$.

\item\textbf{Map 3:} In Map phase 3, mappers once again load-balance tuples
containing vertices or edges as values across the reducers as in Map phases 1
and 2. For each tuples of the form $(r, u, \gdeg{G}{u})$, a mapper emits $n_r$
tuples $(r_1, u, \gdeg{G}{u})$,
$(r_1, u, \gdeg{G}{u})$, \ldots, $(r_{n_r}, u, \gdeg{G}{u})$ -- one for each
reducer. Thus, for each reducer, exactly $n$ tuples containing (full) degree
information are emitted -- one for each vertex of $G$.

\item\textbf{Reduce 3:} In Reduce phase 3, a reducer $r$ now has access to the
degrees of all vertices of $G$ and can thus compute the partition function $F_0$
defined earlier.
Then, for each node $u$ received, a reducer $r$ outputs the tuple $(r, F_0(u), u)$, and for each edge
$(u, v)$ received, a reducer $r$ outputs the tuples $(r, F_0(u), (u, v))$ and
$(r, F_0(v), (u, v))$.

\item In addition to ``packaging'' the vertex and edge information of $G$ so
that incident edges of a node $u$ can be collected at the reducer $F_0(u)$
assigned to simulate computation at $u$, reducers must also emit tuples which
allow both (i) the currently collected degrees of each vertex in $G$ and (ii)
the partition function $F_0$ to be propagated forward through the rounds of the
MapReduce computation. Fortunately this is straightforward: for each degree
tuple $(r, u, \gdeg{G}{u})$ received by reducer $r$, reducer $r$ re-emits the
same tuple. As well, $F_0: V \longrightarrow \{1, \ldots, n_r\}$ can be fully
described by $n$ pairs $(v, F_0(v))$, and so reducer $r$ emits the $n$ tuples
$(r, v, F_0(v))$, which will allow reducer $r$ to ``remember'' the partition
function $F_0(\cdot)$ in the next round. Observe that the totality of the memory
required to support knowledge of the partition function and all degrees in $G$
is $O(n)$, and thus fits into the memory of a reducer without any trouble.

\item\textbf{Map 4:} Finally, in Map phase 4, a mapper receives and processes
two different tuple formats: (i) tuples of the form $(r, r', z)$, where $r'$ is
another reducer index and $z$ is some information (of length $O(1)$ words)
representing either a vertex or an edge; and (ii) tuples of the form
$(r, v, z)$, where $v$ is a vertex identifier and $z$ is either a degree value
or a reducer identifier. In case (i) (tuples of the form $(r, r', z)$), a mapper
emits the tuple $(r', z)$. In case (ii) (tuples of the form $(r, v, z)$, a
mapper simply outputs the same tuple $(r, v, z)$ unchanged.

\item After the Map phase of the round 4 of the MapReduce computation has
completed, the Initialization phase is complete, and the simulation of
$\mathcal{A}_{CC}$ is ready to begin.
\end{itemize}

\noindent\textbf{Simulation} stage. At a high level, a Reduce phase serves as the
``local computation'' phase of the Congested Clique simulation, whereas a Map
phase (together with the subsequent shuffle phase) serves as the
``communication'' phase of the simulation. However, there is, in general, a
constant-factor slow-down because it may be that the sending and receiving of
messages in $\mathcal{A}_{CC}$ could cause the subset of nodes assigned to 
a reducer to
aggregate more than $O(n^{1+\epsilon})$ memory, necessitating a re-partitioning
of the nodes among the reducers so as not to violate the memory-per-machine
constraint. 

Recall that $I_u^{(i)}$ denotes the set of memory addresses used by a node $u$
in round $i$ of $\mathcal{A}_{CC}$.
Let $h_{u, j}^{(i)}$ be the value of word $j \in I_u^{(i)}$
in the memory of node $u$ after node $u$ has completed local computation 
in round $i$ of $\mathcal{A}_{CC}$, but before messages have been sent 
and received in this round. 
For $i > 0$, define a tuple set
\[\mathcal{H}_u^{(i)} = \{(F_i(u), (u, i, h_{u,j}^{(i)})) \colon j \in I_u^{(i)}\}\]
where $F_i(\cdot)$ is the partition function used in round $i$. 
Like $F_0$, defined in the Initialization stage, $F_i$ partitions
$G$ into $n_r$ groups, one per reducer, so that reducer memory constraints
are not violated in round $i$.
The collection of tuples $\mathcal{H}_u^{(i-1)}$ is a representation, in the
MapReduce key-value format, of the information necessary
to simulate the computations of node $u$ in round $i$ of the Congested Clique
algorithm $\mathcal{A}_{CC}$. 
The use of $F_i(u)$ as the key in each of the tuples in $\mathcal{H}_u^{(i)}$ 
ensures that all information needed to simulate a local computation at $u$ in 
$\mathcal{A}_{CC}$ goes to the same reducer.
Additionally, note that the inclusion of the identifier of $u$ with the values allows the
words from $u$'s memory to be reassembled and distinguished from information
associated with other nodes $v \in F_i^{-1}(u)$.
We assume that $\mathcal{H}_u^{(0)}$ is the information in tuple format 
that node $u$ has initially about graph $G$.
In other words, $\mathcal{H}_u^{(0)} = \{(F_0(u), u)\} \cup \{(F_0(u), (u, v)) : v \mbox{ is a
neighbor of } u\}$.

Once an initial partition function $F_0(\cdot)$ has been computed and the initial
collections $\mathcal{H}_u^{(0)}$ have been assembled 
the main goals of our
simulation algorithm are to (i) provide a mechanism for transforming
$\mathcal{H}_u^{(i-1)}$ into $\mathcal{H}_u^{(i)}$ during the reduce phase of a
MapReduce round; and (ii) provide a means of transmitting messages to reducers
of a subsequent round (corresponding to messages transmitted in the Congested
Clique at the end of each round). Since we assume messages to be sent and
received after local computation has occurred during a Congested Clique round,
$\mathcal{M}_u^{(i)}$ can be determined from $\mathcal{H}_u^{(i)}$; in turn,
$\mathcal{H}_u^{(i)}$ is a function of $\mathcal{H}_u^{(i-1)}$ and
$\overline{\mathcal{M}}_u^{(i-1)}$.

We describe the details of the simulation of a single round (round $i$) of a
Congested Clique algorithm $\mathcal{A}_{CC}$ below. Let $j = 3i-1$.
Round $i$ of $\mathcal{A}_{CC}$ is simulated by 
three MapReduce rounds (a total of six Map or Reduce phases) -- Reduce $j-1$, Map $j$, Reduce $j$, Map $j+1$,
Reduce $j+1$, and Map $j+2$.
We assume inductively that as input to Reduce
phase $j-1$ below, each reducer receives, in addition to data tuples, $O(n)$
metadata tuples containing a description of a partition function
$F_{i-1}(\cdot)$ such that for each $r$,
$\sum_{u \in P_r} (|\mathcal{H}_u^{(i-1)}| + |\overline{\mathcal{M}}_u^{(i-1)}|) = O(n^{1+\epsilon})$,
where $P_r = F_{i-1}^{-1}(r)$.

\begin{itemize}
\item\textbf{Reduce phase $j-1$:} In Reduce phase $j-1$, a reducer $r$ receives input
consisting of $\mathcal{H}_u^{(i-1)}$ together with
$\overline{\mathcal{M}}_u^{(i-1)}$ for each $u \in P_r$; for each such $u$,
reducer $r$ performs the following steps:
\begin{itemize}
\item[(i)] Reducer $r$ simulates the local computation of Round $i$ of $\mathcal{A}_{CC}$ at $u$.
\item[(ii)] Reducer $r$ computes $\mathcal{H}_u^{(i)}$ from $\mathcal{H}_u^{(i-1)}$
and $\overline{\mathcal{M}}_u^{(i-1)}$, \textit{but does not yet output} any
tuples of $\mathcal{H}_u^{(i)}$; rather, reducer $r$ outputs only a tuple
$(r, u, s_u)$ containing the size of the information
$s_u = |\mathcal{H}_u^{(i)}|$.
\item[(iii)] Reducer $r$ computes $\mathcal{M}_u^{(i)}$ from $\mathcal{H}_u^{(i)}$,
but again, \textit{does not output} any elements of $\mathcal{M}_u^{(i)}$.
Reducer $r$ then computes, for each $v \in V$, the aggregate count $c_{r,v}$ of
messages emanating from nodes in $P_r$ and destined for $v$, and outputs the
tuple $(r, v, c_{r,v})$.
\item[(iv)] Reducer $r$ outputs the exact same tuples it received as input,
$\mathcal{H}_u^{(i-1)}$ and $\overline{\mathcal{M}}_u^{(i-1)}$. 
\end{itemize}

\item\textbf{Map phase $j$:} Before message tuples can be generated and aggregated (as
a collection $\overline{M}_u^{(i)}$ at reducer $F(u)$) a rebalancing of the
nodes to reducers must be performed to ensure that the reducer-memory constraint
is not violated. In Map phase $j$, a mapper forwards tuples from either a
$\mathcal{H}_u^{(i-1)}$ or a $\overline{\mathcal{M}}_u^{(i-1)}$ through
unchanged. However, for each tuple of the form $(r, u, c_{r,u})$, a mapper
outputs the tuple $(u \bmod n_r, u, c_{r,u})$. 
In addition, for each tuple of the form $(r, u, s_u)$, a
mapper outputs $n_r$ tuples $(r', u, s_u)$ -- one for each reducer $r'$ -- so
that every reducer can know the future size of $\mathcal{H}_u^{(i)}$.

\item\textbf{Reduce phase $j$:} In Reduce phase $j$, a reducer $r$ receives as input
nearly the exact same input (and output) of reducer $r$ in the previous
MapReduce round -- the union of $\mathcal{H}_u^{(i-1)}$ and
$\overline{\mathcal{M}}_u^{(i-1)}$ for each $u \in P_r$ -- except that instead
of receiving tuples of the form $(r, u, c_{r,u})$ for each $u \in V$, reducer
$r$ receives \textit{all} partial message counts for the subset of vertices $u$
for which $u \bmod n_r = r$; as well, each reducer receives $n$ tuples of the
form $(r, u, s_u)$ describing the amount of memory required by node $u$ in round
$i$ of $\mathcal{A}_{CC}$. Reducer $r$ aggregates tuples of the form
$(u \bmod n_r, u, c_{r,u})$ and outputs $(r, u, |\overline{M}_u^{(i)}|)$, since
$|\overline{M}_u^{(i)}|$ is precisely the sum of the partial message counts
$c_{r,u}$. (Notice that a reducer $r$ receives $O(n)$ such tuples.) Reducer $r$
forwards all other tuples through unchanged to the next MapReduce round.

\item\textbf{Map phase $j+1$:} In Map phase $j+1$, a mapper continues to forward all
tuples through unchanged to Reduce phase $j+1$, except that for each tuple of the form
$(r, u, |\overline{M}_u^{(i)}|)$, a mapper outputs $n_r$ tuples
$(r', u, |\overline{M}_u^{(i)}|)$ -- one for each reducer $r'$. In this way,
each reducer in Reducer phase $j+1$ can come to know all $n$ message counts
for each node $u \in V$.

\item\textbf{Reduce phase $j+1$:} In Reduce phase $j+1$, each reducer receives all $n$
message counts (for each node $u \in V$) in addition to the sizes $s_u$ of
the state needed by each node $u$ in round $i$ of $\mathcal{A}_{CC}$. Each
reducer thus has enough information to determine the next partition function
$F_i: V \longrightarrow \{1, \ldots, n_r\}$, defined by
$$F_i(x) = \begin{cases} 1, & \mbox{if } x = 1\\ F_i(x-1),
& \mbox{if } \sum_{v \in L(x)} (s_v + |\overline{M}_v^{(i)}|) \leq n^{1+\epsilon},\\
F_i(x-1)+1, & \mbox{otherwise}
\end{cases}
$$
Here $L(x) = \{v \mid v < x\mbox{ and }F_i(v) = F_i(x-1)\}$.
After determination of the new partition function $F_i$, reducers are now able
to successfully output the ``packaged memory'' $\mathcal{H}_u^{(i)}$ of round
$i$ of $\mathcal{A}_{CC}$, as well as the new messages $m_{u,v}^{(i)}$
sent in round $i$, because the new partition function $F_i$ is specifically
designed to correctly load-balance these tuple sets across the reducers while
satisfying the memory constraint. Therefore:
\begin{itemize}
\item[(i)] Reducer $r$ now simulates the local computation at each $u \in P_r$
and thus outputs the set $\mathcal{H}_u^{(i)}$ (which can be computed from
$\mathcal{H}_u^{(i-1)}$ and $\overline{\mathcal{M}}_u^{(i-1)}$). It is important
to recall here that because mappers operate on key-value pairs one at a time in
the MapReduce model, there is no restriction on the size of the output from any
reducer $r$ in any MapReduce round (other than that it be polynomial).
\cite{KarloffSuriVassilvitskii} Therefore, a reducer $r$ may output (and thus
free-up its memory) each tuple set $\mathcal{H}_u^{(i)}$ as it is created (as
reducer $r$ processes the nodes in $P_r$ one at a time), and so there is no
concern about reducer $r$ attempting to maintain in memory all sets
$\mathcal{H}_u^{(i)}$ for $u \in P_r$ at once. Note that $\mathcal{H}_u^{(i)}$,
as generated by a reducer $r$, should contain tuples of the form
$(r, F_i(u), u, h_{u,l}^{(i)})$ so that mappers in MapReduce round $j+2$ can
correctly deliver $\mathcal{H}_u^{(i)}$ to reducer $F_i(u)$.
Recall that $h_{u, l}^{(i)}$ denotes the contents of the word with address $l$
in node $u$'s memory at the end of local computation in round $i$.
\item[(ii)] As a reducer $r$ processes, and simulates the computation at, each
node $u \in P_r$ one at a time, generating $\mathcal{H}_u^{(i)}$, reducer $r$
also uses $\mathcal{H}_u^{(i)}$ to generate the messages $M_u^{(i)}$ to be sent
by node $u$ in round $i$ of $\mathcal{A}_{CC}$. Reducer $r$ encapsulates
$M_u^{(i)}$ in the tuple set $\mathcal{M}_u^{(i)}$ and outputs it alongside
$\mathcal{H}_u^{(i)}$ before moving on to the next node in $P_r$. As with
$\mathcal{H}_u^{(i)}$, tuples in $\mathcal{M}_u^{(i)}$ should initially be
generated by a reducer $r$ in the form $(r, F_i(v), u, v, m_{u,v}^{(i)})$ so
that mappers in MapReduce round $j+2$ can correctly deliver the set
$\overline{\mathcal{M}}_v^{(i)}$ to reducer $F_i(v)$.
\item[(iii)] Lastly regarding the simulation procedure, whenever a node
$u \in P_r$ being simulated broadcasts a message $b_u^{(i)}$, reducer $r$
outputs the tuple $(r, u, b_u^{(i)})$.
\item[(iv)] After simulation of each node $u \in P_r$ is complete, reducer $r$
also outputs a description of the new partition function $F_i$.
\end{itemize}

\item\textbf{Map $j+2$:} In Map phase $j+2$, a mapper simply transforms the key
in a data tuple as appropriate: for each tuple
$(r, F_i(u), u, h_{u,l}^{(i)})$, a mapper simply emits the tuple
$(F_i(u), u, h_{u,l}^{(i)})$; for each tuple
$(r, F_i(v), u, v, m_{u,v}^{(i)})$, a mapper simply emits the tuple
$(F_i(v), u, v, m_{u,v}^{(i)})$. The exception to this is that tuples
$(r, u, b_u^{(i)})$ containing broadcast messages are expanded: for each, a
mapper emits $n_r$ tuples $(r', u, b_u^{(i)})$ -- one for each reducer $r'$ --
so that every reducer in Reducer phase $j+2$ receives a single copy of each
message broadcast during round $i$ of $\mathcal{A}_{CC}$.

\item Tuples carrying metadata describing the (new) partition function $F_i$ are
forwarded unchanged, because there already exists one copy of each such metadata
tuple for each reducer, and there need be only one such copy per reducer as
well. After Map phase $j+2$, tuples from the sets $\mathcal{H}_u^{(i)}$ and
$\overline{\mathcal{M}}_u^{(i)}$ have been emitted with keys $F_i(u)$, and for
each broadcast message $b_u^{(i)}$, one tuple containing a copy of $b_u^{(i)}$
has been emitted for each reducer as well; thus, in Reduce phase $j+2$,
simulation of round $i+1$ of algorithm $\mathcal{A}_{CC}$ can begin.
\end{itemize}

It remains to comment on the memory-per-machine constraint which must be
satisfied during each MapReduce round. Observe that, inductively, for each $r$,
the sum $\sum_{u \in P_r} (|\mathcal{H}_u^{(i-1)}| + |\overline{\mathcal{M}}_u^{(i-1)}|) = O(n^{1+\epsilon})$.
These data tuples are forwarded unchanged until Reduce phase $j+1$, in which the
new partition function $F_i(\cdot)$ for the next round of simulation is
computed, and then collectively $\mathcal{H}_u^{(i-1)}$ and
$\overline{\mathcal{M}}_u^{(i-1)}$ are transformed into $\mathcal{H}_u^{(i)}$
and $\mathcal{M}_u^{(i)}$. By construction of the partition functions $F_{i-1}$
and $F_i$, and by the assumption that $\mathcal{A}_{CC}$ is a
$(n^{1+c}, n^{1+\epsilon})$-\textit{lightweight} algorithm, it follows that
these data tuples are never present on any reducer a number that exceeds
$\Theta(n^{1+\epsilon})$. Secondly, it should be mentioned that because
broadcast messages are not duplicated at any reducer $r$, no reducer will ever
receive more than $n = O(n^{1+\epsilon})$ tuples containing broadcast messages.
Thirdly, tuples containing state or message counts are never present in a number
exceeding $n$ at any reducer, and \textit{partial} message counts are explicitly
load-balanced so that only $O(n)$ such information is passed to a single reducer
as well. Finally, metadata tuples describing a partition function never exceed
$\Theta(n)$ on any reducer because the domain of each partition function has
size $n$.
\end{proof}

\section{Coloring in the MapReduce Framework}
\label{sect:MapReduceFramework}

Using the simulation theorem of Section \ref{sect:MRfromCC}, we can simulate
Algorithm \textsc{HighDegCol} in the MapReduce model and thereby achieve an
$O(\Delta)$-coloring MapReduce algorithm running in expected-$O(1)$ rounds. As
in Lattanzi et al.~\cite{LattanziFiltering}, we consider graphs with
$\Omega(n^{1+c})$ edges, $c > 0$.

\begin{theorem}
When the input graph $G$ has $\Omega(n^{1+c})$ edges, and $0 \leq \epsilon < c$, there exists an
$O(\Delta)$-coloring algorithm running in the MapReduce model with
$\Theta(n^{c-\epsilon})$ machines and $\Theta(n^{1+\epsilon})$ memory per
machine, and having an expected running time of $O(1)$ rounds.
\label{theorem:DistColorMR}
\end{theorem}
\begin{proof}
It is easy to examine the lines of code in Algorithm \textsc{HighDegCol} to ascertain
that the total amount of non-broadcast communication in any round in bounded above by
$O(n^{1+c})$.
Specifically, the total non-broadcast communication corresponding to only two lines of code
-- Lines 6 and 11 -- can be as high as $\Theta(n^{1+c})$. For all other lines of code,
the volume of total non-broadcast communication is bounded by $O(n)$.
Similarly, it is easy to examine the lines of code in Algorithm \textsc{HighDegCol} to verify
that the total memory (in words) used by all nodes for their local computations in any one round
is bounded above by $O(n^{1+c})$.
Finally, it is also easy to verify that the maximum amount of memory used by a node in any round of computation is $O(n)$.

Thus, Algorithm \textsc{HighDegCol} is an $(n^{1+c}, n)$-lightweight algorithm on a Congested Clique
and applying the Simulation Theorem (Theorem \ref{theorem:CCColoring}) to this algorithm yields
the claimed result.
\end{proof}

\noindent
It is worth emphasizing that the result holds even when $\epsilon = 0$; in other
words, even when the per machine memory is $O(n)$, the algorithm can compute an
$O(\Delta)$-coloring in $O(1)$ rounds. This is in contrast with the results in
Lattanzi et al.~\cite{LattanziFiltering}, where $O(1)$-round algorithms were
obtained (e.g., for maximal matching) with $n^{1 + \epsilon}$ per machine memory, only when $\epsilon > 0$.
In their work, setting $\epsilon = 0$ (i.e., using $\Theta(n)$ memory per machine) 
resulted in $O(\log n)$ round algorithms.

We end with the following corollary that is an immediate consequence of Theorem \ref{theorem:DistColorMR}.

\begin{corollary}
The problem of computing an $O(\Delta)$-coloring for an $n$-node graph with
maximum degree $\Delta$ and at least $\Omega(n^{1+c})$ edges, for $c > 0$ is in
$\mathcal{MRC}^0$.
\end{corollary}

\section{Conclusions}
\label{sect:Conclusions}

The results in this paper connect two models that are usually studied by
different research communities. In general, it would be interesting to see if this
connection has benefits beyond those discussed in the paper. 
Also, it would be be interesting to study differences between these two models.
For example, the Congested Clique model allows nodes to remember arbitrary amount of 
information from one round to the next. Does this give the Congested Clique model
a provable advantage over the ``stateless'' MapReduce model?

For the ``small $\Delta$'' case, i.e., when $\Delta = O(\mbox{poly}(\log n))$, our paper presents an $O(\llln)$-round $(\Delta+1)$-coloring algorithm
on a Congested Clique.
One question that interests us is whether $O(1)$ rounds will suffice to
compute an $O(\Delta)$-coloring even when $\Delta$ is small?

Following the lead of Lattanzi et al.~\cite{LattanziFiltering}, we have assumed that each machine in the MapReduce
model contains at least $\Omega(n)$ memory for processing an $n$-node graph.
Relaxing this assumption is interesting and leads to the question of 
whether for some $\epsilon > 0$, $O(1)$ MapReduce rounds would suffice to compute an
$O(\Delta)$-coloring, even when the per machine memory is
$O(n^{1-\epsilon})$.

\bibliographystyle{plain}

\end{document}